\documentclass[11pt,a4paper,reqno]{amsproc}

\makeatletter
\usepackage{datetime}
\usepackage[pdfborder={0 0 0}]{hyperref}
\usepackage{upref,cases}
\hypersetup{citebordercolor=.1 .6 1}

\usepackage{geometry}
\IfFileExists{mymtpro2.sty}{\usepackage[subscriptcorrection]{mymtpro2}

  \def\cap{\capprod}
  \def\bigcup{\bigcupprod}
  
  \def\bigcupdisjoint{\mathop{\kern10pt\raisebox{4pt}{$\cdot$}\kern-12pt\bigcup}\limits}
}%
	{\usepackage{amssymb,bm,dsfont}

}

\numberwithin{equation}{section}
\newtheoremstyle{ttheorem}%
       {1.8ex\@plus1ex}                
       {2.1ex\@plus1ex\@minus.5ex}      
       {\itshape}           
       {0pt}                   
       {\bfseries}          
       {.}                  
       {.5em}               
       {}                

\newtheoremstyle{ddefinition}%
       {1.8ex\@plus1ex}                
       {2.1ex\@plus1ex\@minus.5ex}      
       {}           
       {0pt}                   
       {\bfseries}           
       {.}                  
       {.5em}               
       {}                

\newtheoremstyle{rremark}%
       {1.8ex\@plus1ex}                
       {2.1ex\@plus1ex\@minus.5ex}      
       {\normalfont}        
       {0pt}                   
       {\bfseries}           
       {.}                  
       {.5em}               
       {}                   

\theoremstyle{ttheorem}
\newtheorem{theorem}{Theorem}[section]
\newtheorem{lemma}[theorem]{Lemma}

\newtheorem{corollary}[theorem]{Corollary}

\theoremstyle{ddefinition}
\newtheorem{definition}[theorem]{Definition}

\theoremstyle{rremark}
\newtheorem{remark}[theorem]{Remark}
\newtheorem{myremarks}[theorem]{Remarks}
\newtheorem{myexamples}[theorem]{Examples}

\newenvironment{remarks}{\begin{myremarks}\begin{nummer}}%
    {\end{nummer}\end{myremarks}}
    {\end{nummer}\end{myexamples}}

\newcounter{numcount}
\newcommand{\labelnummer}{(\roman{numcount})}%

\makeatletter
\providecommand{\showkeyslabelformat}[1]{\relax}        
\let\mysaveformat\showkeyslabelformat                   %
\def\myformat#1{\raisebox{-1.5ex}{\mysaveformat{#1}}}   %

\newenvironment{nummer}%
  {\let\curlabelspeicher\@currentlabel%
    \begin{list}{\textup{\labelnummer}}%
      {\usecounter{numcount}\leftmargin0pt%
        \topsep0.5ex\partopsep2ex\parsep0pt\itemsep0ex\@plus1\p@%
        \labelwidth2.5em\itemindent3.5em\labelsep1em%
      }%
    \let\saveitem\item%
    \def\item{\saveitem%
      \def\@currentlabel{\curlabelspeicher\kern.1em\labelnummer}}%
    \let\savelabel\label%
    \def\label##1{{\ifnum\thenumcount=1\let\showkeyslabelformat\myformat\fi\savelabel{##1}}%
										{\def\@currentlabel{\labelnummer}%
									 	\let\showkeyslabelformat\@gobble
									 	\savelabel{##1item}%
										}%
	   							}%
  }{\end{list}}%

  {\let\curlabelspeicher\@currentlabel%
    \begin{list}{\textup{\labelnummer}}%
      {\usecounter{numcount}\leftmargin0pt%
        \topsep0.5ex\partopsep2ex\parsep0pt\itemsep0ex\@plus1\p@%
        \labelwidth2.5em\itemindent0em\labelsep1em%
        \leftmargin2.5em}%
    \let\saveitem\item%
    \def\item{\saveitem%
      \def\@currentlabel{\curlabelspeicher\kern.1em\labelnummer}}%
    \let\savelabel\label%
    \def\label##1{{\ifnum\thenumcount=1\let\showkeyslabelformat\myformat\fi\savelabel{##1}}%
										{\def\@currentlabel{\labelnummer}%
									 	\let\showkeyslabelformat\@gobble
									 	\savelabel{##1item}%
										}%
    							}%
  }{\end{list}}%

\def\section{\@startsection{section}{1}%
  \z@{1.3\linespacing\@plus\linespacing}{.5\linespacing}%
  {\normalfont\bfseries\centering}}

\def\subsection{\@startsection{subsection}{2}%
  \z@{.8\linespacing\@plus.5\linespacing}{-1em}%
  {\normalfont\bfseries}}

\def\nlsubsection{\@startsection{subsection}{2}%
  \z@{.8\linespacing\@plus.5\linespacing}{.1ex}%
  {\normalfont\bfseries}}

\let\@afterindenttrue\@afterindentfalse%

\renewenvironment{proof}[1][\proofname]{\par \normalfont
  \topsep6\p@\@plus6\p@ \trivlist 
  \item[\hskip\labelsep\scshape
    #1\@addpunct{.}]\ignorespaces
}{%
  \qed\endtrivlist
}

\def\ps@firstpage{\ps@plain
  \def\@oddfoot{\normalfont\scriptsize \hfil\thepage\hfil
     \global\topskip\normaltopskip}%
  \let\@evenfoot\@oddfoot
  \def\@oddhead{
    \begin{minipage}{\textwidth}
      \normalfont\scriptsize
      \emph{\insertfirsthead}
    \end{minipage}}
  \let\@evenhead\@oddhead 
}

\def\insertfirsthead{}

\def\@cite#1#2{{%
 \m@th\upshape\mdseries[{#1}{\if@tempswa, #2\fi}]}}
\renewcommand{\H}{\mathcal{H}}

\newcommand{\C}{\mathbb{C}}
\newcommand{\N}{\mathbb{N}}

\newcommand{\R}{\mathbb{R}}

\renewcommand{\le}{\leqslant}
\renewcommand{\ge}{\geqslant}

\DeclareMathOperator{\tr}{tr}

\providecommand{\wtilde}[1]{\widetilde{#1}}

\providecommand{\bigcupdisjoint}{\mathop{\kern7pt\raisebox{6pt}{$\cdot$}\kern-9.5pt\bigcup}\limits}

\let\<\langle
\let\>\rangle

\providecommand{\norm}[1]{\lVert#1\rVert}

\providecommand{\bignorm}[1]{\bigl\lVert#1\bigr\rVert}
\providecommand{\Bignorm}[1]{\Bigl\lVert#1\Bigr\rVert}

\providecommand{\Bigparens}[1]{\Bigl(#1\Bigr)}

\newcommand{\Oh}{\mathrm{O}}
\newcommand{\oh}{\mathrm{o}}

\newcommand{\1}{1}

\newcommand{\upd}{\mathrm{d}}
\renewcommand{\d}{\upd}   

\newcommand{\hairspace}{\kern .04167em}

\renewcommand{\S}{\mathcal{S}}

\def\clap#1{\hbox to 0pt{\hss#1\hss}}

\def\bra{\makeatletter\@ifstar\@bra\@@bra}
\def\@bra#1{\hairspace #1\>}
\def\@@bra#1{\lvert\@bra{#1}}
\def\ket{\makeatletter\@ifstar\@ket\@@ket}
\def\@ket#1{\<#1\hairspace}
\def\@@ket#1{\@ket{#1}\rvert}

\makeatother
\begin{document}


\title[Anderson's Orthogonality Catastrophe for Dirac-$\delta$]
{The asymptotics of an eigenfunction-correlation determinant for Dirac-$\delta$ perturbations \\
(Anderson's Orthogonality Catastrophe for Dirac-$\delta$)
}

\author[M.\ Gebert]{Martin Gebert}

\address{Mathematisches Institut,
  Ludwig-Maximilians-Universit\"at M\"unchen,
  Theresienstra\ss{e} 39,
  80333 M\"unchen, Germany}

\email{gebert@math.lmu.de}

\thanks{Work supported by SFB/TR 12 of the German Research Council
(DFG)}


\begin{abstract}
We give a proof of the exact asymptotic behaviour in Anderson's Orthogonality Catastrophe for Dirac-$\delta$ perturbations. 
 We provide the asymptotics
 of the scalar product of the ground states  
 of two non-interacting Fermi gases confined to a $3$-dimensional ball $B_L$
 of radius $L$ 
 in the thermodynamic limit, where the underlying
 one-particle operators differ by 
 a Dirac-$\delta$ perturbation. More precisely,
 we show the algebraic decay of the correlation determinant
 $\big|\det\big(\langle\varphi_j^L, \psi_k^L\rangle\big)_{j,k=1,...,N}\big|^2= L^{-\zeta(E)+ \text{o}(1)}$, 
as $N,L\to\infty$ and $N/|B_L|\to~ \rho>0$, where
$\varphi_j^L$ and $\psi_k^L$ denote the lowest-energy eigenfunctions of the finite-volume one-particle
Schr\"odinger operators.
The decay exponent is given in terms of the s-wave scattering phase shift 
 $\zeta(E):=\frac 1 {\pi^2}\delta^2(\sqrt E)$. 
 For an attractive Dirac-$\delta$ perturbation 
 we conclude that the decay exponent $ \frac 1 {\pi^2}\Vert\arcsin |T(E)/2|\Vert^2_{\text{HS}}$ 
 found in \cite{GKM2} does not provide a sharp upper 
 bound on the decay of the correlation determinant.
\end{abstract}

\maketitle

\section{Introduction}

We consider the asymptotics of the scalar product of the ground states of two non-interacting finite-volume
$N$-particle Schr\"odinger operators in the thermodynamic limit approaching the
particle density $\rho(E)>0$ corresponding to the Fermi energy $E>0$.
Here, the underlying one-particle Schr\"odinger operators are the negative Laplacian 
in $3$-dimensional Euclidean space and the negative Laplacian with a Dirac-$\delta$ or zero-range 
perturbation located at the origin. 
We restrict this pair to the ball $B_L(0)$
of radius $L$ and are interested in the $L$-asymptotics of the scalar product 
of the ground states of the corresponding two non-interacting $N$-particle operators,
which we call the ground-state overlap in the sequel.
Using the representation of the ground states as Slater determinants,
we see that
the ground-state overlap is the following correlation determinant
\begin{equation}\label{intro2}
 \S^N_L:=\det\Bigparens{\big\<\varphi_j^L, \psi_k^L\big\>}_{1\le j,k\le N}.
\end{equation}
In this note, we are interested in its thermodynamic limit, i.e. 
increasing $L$ and $N\in\N$ simultaneously such that
$N/|B_L(0)|\to\rho(E)>0$, where $\rho(E)$ denotes the integrated density of states of the negative 
Laplacian at the energy $E>0$.
Here, $\varphi_j^L$ and $\psi_k^L$ are the normalized eigenfunctions belonging to the $N$ lowest eigenvalues 
of the restricted operators, 
which we call $H_L$ and $H_{\alpha,L}$, and $\<\cdot,\cdot\>$ denotes the scalar product in $L^2(B_L(0))$. 
Since we restrict the operators to the ball $B_L(0)$, the operators $H_L$ and $H_{\alpha,L}$ admit a decomposition with respect to angular momentum. The operators differ in this decomposition in the lowest angular momentum channel only and we choose the same eigenfunctions of $H_L$ and $H_{\alpha,L}$ in all angular momentum channels $\ell\ge 1$. Thus, the problem reduces effectively to a problem on the half axis. 
Anderson claimed in \cite{PhysRev.164.352} that in the case of a Dirac-$\delta$-perturbation the determinant
admits the asymptotics
\begin{equation}\label{intro}
 \big|\S^N_L \big|^2\sim L^{-\zeta(E)}
\end{equation}
as $N,L\to\infty$, $N/|B_L(0)|\to\rho(E)>0$, 
where 
\begin{equation}
\zeta(E):=\frac{1}{\pi^2}\delta^2(\sqrt E)
\end{equation}
and $\delta$ refers to the s-wave scattering phase shift.
This algebraic decay of the ground-state overlap is called
Anderson's orthogonality catastrophe in the physics literature  
and we refer to \cite{GKM} for further references. 

The starting point of the proofs of previous rigorous results is the following expansion of the determinant
\begin{equation}\label{11111}
 \ln \big|\S_L^N \big|^2= -\sum_{n=1}^\infty 
 \frac 1 n \tr\left\{ \left( \1_{(-\infty,\lambda_N^L]}(H_L)\1_{[\mu_{N+1}^L,\infty)}(H_{\alpha,L})\right)^n\right\},
\end{equation}
valid for appropriate choices of $N$,
where $\lambda_N^L$ and $\mu_{N+1}^L$ denote the $N$th and $(N+~1)$th eigenvalue of
the finite-volume operators $H_L$ and $H_{\alpha,L}$, see \cite{GKM2}. 
Thus, estimates on the correlation determinant $S_L^N$ are closely 
related to asymptotics of products of spectral projections given in \eqref{11111}.
Considering only the $n=1$ term in \eqref{11111}, 
the first rigorous bounds on $S^N_L$ were proved in \cite{KuOtSp13} valid for one-dimensional systems and
short-range perturbations. 
They found the upper bound $|\S_L^N|^2\lesssim L^{-\wtilde \gamma}$
with the decay exponent $\wtilde{\gamma}(E):=\frac 1 {\pi^2}\norm{T(E)/2}^2_{\text{HS}}$, where
$T$ refers to the scattering $T$-matrix of the corresponding infinite-volume operators, 
and a non-optimal lower bound.
Later in \cite{GKM} the same upper bound $\wtilde{\gamma}(E):=\frac 1 {\pi^2}\norm{T(E)/2}^2_{\text{HS}}$ 
was deduced for quite general pairs of Schr\"odinger operators in arbitrary dimension, which differ by a sign-definite
potential.
Taking all summands in \eqref{11111} into account, \cite{GKM2} proved an upper bound
with the decay exponent 
\begin{equation}
\gamma(E):=\frac 1 {\pi^2} \norm{\arcsin |T(E)/2|}^2_{\text{HS}}
\end{equation}
in the general setting discussed in \cite{GKM}. 
Let us point out that these previous results concern upper 
bounds and are also valid for special choices of thermodynamic limits only.

Here, in the toy model of a Dirac-$\delta$ perturbation we provide the exact asymptotics
of the correlation determinant and we consider arbitrary thermodynamic limits 
approaching a particle density $\rho>0$, see Theorem \ref{main:thm:3d} below.
We show this using a representation of the ground-state overlap other than \eqref{11111}, which is
valid for rank-1-perturbations, i.e.
\begin{equation}
    \left|  \S_L^N\right|^2
    =
    \prod_{j=1}^N\prod_{k=N+1}^\infty\frac{\big|\mu^L_k-\lambda^L_j\big|\big|\lambda^L_k-\mu^L_j\big|}
    {\big|\lambda^L_k-\lambda^L_j\big|\big|\mu^L_k-\mu^L_j\big|},
 \end{equation}
 where $\lambda_k^L$ and $\mu_j^L$ are the eigenvalues of the pair of the finite-volume Schr\"odinger
 operators, see Section \ref{sec:prod}. This formula is known in physics literature and
 goes back at least to \cite{Tanabe}.
Using the latter formula, we give a straightforward proof of the algebraic decay \eqref{intro} with
the exponent $\zeta(E)=\frac{1}{\pi^2}\delta^2(\sqrt E)$, as Anderson predicted.
It turns out that the decay exponent is equal to the one found in \cite{GKM2} 
in the case of 
a repulsive Dirac-$\delta$ perturbation 
only, i.e. $\zeta(E) = \gamma(E)$.
On the other hand, we obtain 
$\zeta(E)> \gamma(E)$ 
for an attractive Dirac-$\delta$
see Remark \ref{1.1} below.
Hence, the decay exponent $\gamma(E)$ does not provide the exact asymptotics of \eqref{intro2}.
We conjecture that the exponent $\gamma(E)$ gives the correct decay exponent only whenever the appropriately continuously normalized phase shifts do not exceed $|\pi/2|$.

Recently, \cite{magnetic} proved the asymptotics of a shifted correlation determinant
for one-dimensional models with a perturbation by a magnetic field.
A related problem, which we mention for completeness, is considering the asymptotics of products of spectral 
projections of infinite-volume operators, similar to \eqref{11111}. 
This was done in the proof of \cite{GKM2} and extended in \cite{FrankPush}.

\section{Model and results}

We start with the operator $-\Delta_0:~C_c^\infty\big(\R^3\backslash \{0\}\big)\to~ L^2(\R^3)$,
which has deficiency indices $(1,1)$. Therefore, $-\Delta_0$ gives rise to a one-parameter
family of self-adjoint extensions which we index by $\alpha\in\R$ and denote by $-\Delta_\alpha$, 
see \cite[Chapter 1]{zbMATH02132167}.
We refer to 
$-\Delta_{\alpha}$ as
the negative Laplacian with a Dirac-$\delta$ perturbation sitting at the origin $0$ of strength $\alpha$.
Throughout, we consider for $\alpha\in\R$ the pair of Schr\"odinger operators 
\begin{equation}
 H:=-\Delta\qquad \text{and} \qquad H_{\alpha}:= -\Delta_{\alpha}
\end{equation}
on the Hilbert space $\H=L^2(\R^3)$,
where $-\Delta$ is the negative Laplacian.
More precisely, following \cite[Chapter 1]{zbMATH02132167}, the operators $H$ and $H_\alpha$ admit a decomposition with
respect to angular momentum. Thus, there exists a unitary $U$ such that
both operators transform into the direct sum 
\begin{equation}\label{3-d:eq1}
 UHU^*=\bigoplus_{\ell\in\N_0}\bigoplus_{-\ell\le m \le \ell} h^\ell
 \quad \text{and} \quad
 UH_\alpha U^* = \bigoplus_{\ell\in\N_0}\bigoplus_{-\ell\le m \le \ell} h_\alpha^\ell,
 \end{equation}
where $h^\ell_{(\alpha)}: L^2((0,\infty))\supset \text{dom}(h^\ell_{(\alpha)})\to L^2((0,\infty))$ and
$h^\ell=h_\alpha^\ell$ for all $\ell\ge 1$. In the $\ell=0$ case the operators are given by
\begin{align}
h^0=-\frac{\d^2}{\d x^2},\quad
\text{dom}(h^0)=  \big\{& f\in L^2((0,\infty)):\, f,f'\in AC_{\text{loc}}((0,\infty)); \\
			  & f(0+) =0;\, f''\in L^2((0,\infty))\big\}\nonumber
\end{align}
\begin{align}
 \ \ h^0_\alpha=-\frac{\d^2}{\d x^2},\quad 
 \text{dom}(h^0_\alpha)=\big\{& f\in L^2((0,\infty)):\, f,f'\in AC_{\text{loc}}((0,\infty)); \\
			  &-4\pi\alpha f(0+)+f'(0+)=0;\, f''\in L^2((0,\infty))\big\},\nonumber
\end{align}
where we denote by $AC_{\text{loc}}((0,\infty))$ the set of all locally absolutely continuous
functions.  
Thus, the difference of $H$ and $H_\alpha$ takes place in the lowest angular momentum 
channel via a different boundary condition at $0$ which we parametrise by $\alpha\in\R$. 
In the following we are interested in 
the restrictions of these operators to the ball $B_L(0)$ of radius $L$ around 
the origin
\begin{equation}
 H_L:= -\Delta_L\qquad \text{and} \qquad H_{\alpha,L}:= -\Delta_{\alpha,L}.
\end{equation}
Here, $-\Delta_L$ is the negative Dirichlet Laplacian on $B_L(0)$.
The operator $-\Delta_{\alpha,L}$ corresponds to the restriction of the operator $-\Delta_\alpha$ 
 imposing Dirichlet boundary condition at $L$ in each angular momentum channel, i.e.
 also the restriction of $-\Delta_\alpha$ to $B_L(0)$ with Dirichlet boundary conditions.
Thus, $H_L$ and $H_{\alpha,L}$ differ as well as before in the lowest angular momentum channel
only by a different boundary condition at $0$. 
We call the corresponding operators in the $\ell=0$ channel, i.e the restrictions of
$h^0$ and $h^0_\alpha$ to the interval $(0,L)$ with Dirichlet boundary condition at $L$,
\begin{equation}
 h^0_L\qquad \text{and} \qquad h^0_{\alpha,L}.
\end{equation}
Using standard results for regular Sturm-Liouville operators, 
 we obtain for all $z\in \varrho(h^0_L)\cap \varrho(h_{\alpha,L}^0)$ a vector 
$\eta_{L,z}^\alpha\in\ L^2(B_L(0))$ such that the resolvents satisfy
\begin{equation}\label{delta:lemma:rank-1} 
 \frac 1 {h^0_L-z} - \frac 1 {h^0_{\alpha,L}-z} = \big|\eta_{L,z}^\alpha\big\>\big\<\eta_{L,z}^\alpha\big|.
\end{equation}
Throughout, we write $\varrho(A)$ for the resolvent set of an operator $A$. 
Thus, $h^0_{\alpha,L}$ is a rank-1-perturbation of $h^0_L$ in the resolvent, and the same is true for
the pair
$H_{\alpha,L}$ and $H_L$.  
We point out that the perturbation is not compactly supported since $\eta_{L,z}^\alpha$ is $L$ dependent.
Moreover, the compactness of the resolvents
of $H_L$ and $H_{\alpha,L}$ imply that both $H_L$ and $H_{\alpha,L}$ have discrete spectra.
We write 
\begin{equation}
\lambda_1^L\le\lambda_2^L\le\dotsb\quad \text{and} \quad
\mu_1^L\le\mu_2^L\le\dotsb
\end{equation}
for their non-decreasing sequences of eigenvalues, counting multiplicities, and $(\varphi_j^L)_{j\in\N}$ and
$(\psi_k^L)_{k\in\N}$ for the corresponding
sequences of normalized eigenfunctions, where we choose the same eigenvectors 
for $H_L$ and $H_{\alpha,L}$
in any angular momentum 
channel $\ell\ge 1$. 
This choice ensures that the eigenfunctions of $H_L$ and $H_{\alpha.L}$ differ in the lowest angular
momentum channel only and our problem reduces to a problem on the half axis. 
Let us point out that in the case of $\alpha<0$ there exists precisely one negative eigenvalue $\mu_1=-(4\pi\alpha)^2$
for the infinite-volume operator $H_\alpha$, respectively $h^0_{\alpha}$, see  \cite[Chapter 1]{zbMATH02132167}. 
Dirichlet-Neumann bracketing implies  
$h^0_{\alpha} \le h^0_{\alpha.L}\oplus h^0_{L^c}$, where $h^0_{L^c}$ denotes the negative Laplacian on
$(L,\infty)$ with Dirichlet boundary condition at $L$. Thus, in the case of $\alpha<0$
we obtain the uniform lower bound on the finite-volume operators
\begin{equation}\label{delta:unif:below}
 H_{\alpha,L}\ge -(4\pi\alpha)^2\quad\text{and equivalently}\quad  h^0_{\alpha,L}\ge -(4\pi\alpha)^2.
\end{equation}
Let $N\in\N$.
In the following we are interested in the correlation determinant
\begin{equation}
	\S^N_L := 
    \det\Bigparens{\big\<\varphi_j^L, \psi_k^L\big\>}_{1\le j,k\le N}.
  \label{def:overlap}
\end{equation}
The main result concerning $\S^N_L$ is the following.

\begin{theorem}\label{main:thm:3d}
Let $\alpha\in\R$, $E>0$ and 
$N_{(\,\cdot\,)}(E):\R_+\to\N$ an arbitrary function subject to
\begin{equation}\label{XYZ}
\lim_{L\to\infty}  \frac{N_L(E)}{|B_L(0)|}= \rho(E):= \frac{E^{3/2}}{8\pi^3},
\end{equation}
where
$\rho$ denotes the integrated density of states of the operator $-\Delta$.
Then, the correlation determinant corresponding
to the pair $H_L$ and $H_{\alpha,L}$ admits the asymptotics
		\begin{equation}\label{111112}
		 \big|\S^{N_L(E)}_L\big|^2 = L^{-\frac 1 {\pi^2}\delta_{\alpha}^2(\sqrt E)+ \oh(1)},
               \end{equation}
as $L\to\infty$,
equivalently,
\begin{equation}
 \lim_{L\to\infty} \frac{\ln\big|\S^{N_L(E)}_L\big|^2}{\ln L} = -\frac 1 {\pi^2}\delta_{\alpha}^2(\sqrt E),
\end{equation}
and $\delta_{\alpha}$ is given by Definition \ref{delta:def:phase} below. 
\end{theorem}

\begin{definition}[Scattering phase shift]\label{delta:def:phase}
 Let $k>0$. Then,  the scattering phase shift is defined by
 \begin{align}
 \delta_\alpha(k):=
  \begin{cases}
    \arctan\left(\frac k {4\pi \alpha}\right)\ \qquad\ \ & \text{for}\ \alpha\ge 0\\
    \pi - \arctan\big(\frac k {4\pi |\alpha|}\big)\ \ & \text{for}\ \alpha\le 0,
  \end{cases}
 \end{align}
 where we use the convention $\arctan\big(\frac k 0\big):=\frac\pi 2$ for $k>0$.
\end{definition}

\begin{remark}\label{1.0}
 We choose the same eigenfunctions in the $l\ge 1$ angular momentum channels because we are considering an s-wave scattering problem. In principle, this choice is only necessary if $\lambda_N^L$ is degenerate and $\S^N_L$ takes only a proper subset of the eigenfunctions in the $\lambda_N^L$ eigenspace into account. 
\end{remark}

\begin{remarks}\label{1.1}
\item 
 The separate definitions of the phase shift
 are reminiscent to the existence of a negative eigenvalue whenever $\alpha<0$
 and Levinson's theorem. 
 \item
 Due to the nature of a Dirac-$\delta$ perturbation in $3$ dimensions the same result is
 apparently valid for the corresponding problem on the half-axis. 
 \item
 We emphasise that we allow arbitrary thermodynamic limits approaching the particle density $\rho>0$. 
 In other words, given a Fermi energy $E>0$ we allow arbitrary choices of the particle number satisfying equation
 \eqref{XYZ}.
 \item
 The $\oh(1)$-error in \eqref{111112} depends on the particular choice of the thermodynamic limit.
 To see this, we refer to equations \eqref{1.} and \eqref{2.} in the proof of Theorem \ref{main:thm:3d}. 
 In particular, we think that the error cannot be improved allowing arbitrary thermodynamic limits.
 \item
\cite{Affleck199735} and references therein suggest that there is a connection between the exact decay exponent in Anderson's orthogonality and the so-called finite-size energy for systems on the half axis. This is the $\Oh(1/L)$ correction in the difference of the ground-state energies. It was shown in \cite{MG:delta} that this correction depends, in contrast to the above result, on the precise thermodynamic limit. Hence, we question a fundamental connection of these quantities. 
 \item
 In \cite[Theorem 2.2]{GKM2} an upper bound on the ground-state overlap is proved 
 for quite general pairs of Schr\"odinger operators which is valid
 for subsequences only. More precisely, they prove for a subsequence
 \begin{equation}
  \limsup_{L\to\infty}\frac {\ln\big|\S^{N_L(E)}_L\big|}{\ln L} \le - \frac{\gamma(E)}2,
 \end{equation}
 where 
 \begin{equation}
  \gamma(E):=\frac 1 {\pi^2} \norm{\arcsin|T(E)/2|}^2_{\text{HS}}
 \end{equation}
 and $T$ denotes the scattering $T$-matrix.
 Since we consider here s-wave scattering, 
 we restrict ourselves to the lowest angular momentum channel. In this case, $T(E)$ is a complex number and
 $|T(E)/2|=\sin(\delta_\alpha(\sqrt E))$. Now, computing $\gamma(E)$ yields
 \begin{align}
  \gamma(E)
  = 
  \begin{cases}
   \frac 1 {\pi^2}\delta_\alpha^2(\sqrt E) &\ \text{for}\ \delta_\alpha(\sqrt E)\le \frac\pi 2\\
   \frac 1 {\pi^2} \big(\delta_\alpha(\sqrt E)-\pi\big)^2 &\ \text{for}\ \delta_\alpha(\sqrt E)\ge \frac\pi 2.
  \end{cases}
 \end{align}
Thus, in general the decay exponent $\gamma(E)$
does not provide a sharp upper bound on the correlation determinant whenever the phase shift
is bigger than $ \pi/ 2$. In our model this is equivalent to $\alpha< 0$
which we refer to as the attractive case. We warn the reader that, generally, the modulus of the phase shift may be bigger than $|\pi/2|$ even in cases where the perturbation does not create a bound state. However, if the perturbation creates bound states, Levinson's theorem will immediately give a phase shift strictly bigger than $|\pi/2|$ for Fermi energies near $0$. 
\end{remarks}

The  proof of Theorem \ref{main:thm:3d} follows from a different approach 
than the one made in \cite{GKM} and \cite{GKM2}, i.e. we do not use the representation
\eqref{11111} in this article. Here, the key is the
following remarkable product representation of the determinant in terms of the
eigenvalues of the finite-volume Schr\"odinger operators. To our knowledge, this was first stated in \cite{Tanabe}.

\begin{lemma}
Let $N\in\N$. Then,
  \begin{equation}
    \Big|  \det\Bigparens{\big\<\varphi_j^L, \psi_k^L\big\>}_{1\le j,k\le N}\Big|^2
    =\,
    \prod_{j=1}^N\prod_{k=N+1}^\infty\frac{\big|\mu^L_k-\lambda^L_j\big|\big|\lambda^L_k-\mu^L_j\big|}
    {\big|\lambda^L_k-\lambda^L_j\big|\big|\mu^L_k-\mu^L_j\big|}.
 \end{equation}
\end{lemma}

We start with proving this product representation for general pairs of compact  operators which differ by a 
rank-1-perturbation in Section \ref{sec:prod}. 
We apply this to our setting in Section \ref{proof:thm:main} and prove Theorem  \ref{main:thm:3d}.

\section{Representation of the ground-state overlap}\label{sec:prod}

In this section we prove a quite general representation for determinants of eigenvectors 
of pairs of operators which differ by a rank-1-perturbation. 
The main result in this section, Theorem \ref{prod:thm1}, will be the key to the proof of 
Theorem \ref{main:thm:3d}.

Let $\H$ be a separable infinite-dimensional Hilbert space and $A:\H\to \H$ be a compact, linear and self-adjoint operator.
Moreover, we assume $A \ge 0$
with $\text{ker}(A)=\{0\}$. We define
\begin{equation}\label{determinant:def1}
 B:=A+|\phi\rangle \langle \phi|
\end{equation}
for some $0\neq\phi\in\H$.
We write $\alpha_1\ge \alpha_2\ge \cdots$ and $\beta_1\ge\beta_2\ge\cdots$ for the non-increasing sequences 
of eigenvalues of $A$, respectively $B$
and
denote by $\left(\varphi_j\right)_{j\in\N}$ and $\left(\psi_k\right)_{k\in\N}$ the corresponding normalized eigenvectors. Since $A$ and $B$ differ by a rank-1-perturbation, the min-max theorem
implies that the eigenvalues interlace.
For simplicity we assume in addition the following strict interlacing condition 
\begin{equation}\label{determinant:assumption}
 \beta_1>\alpha_1>\beta_2>\alpha_2>\cdots.
\end{equation}
In particular, $\beta_k\neq\alpha_j$ for all $j,k\in \N$.
Furthermore, the above implies cyclicity of $\phi$.
Assumption \eqref{determinant:assumption} is not necessary but simplifies notation and computations. In the general case one has to consider the restriction to the cyclic subspace generated by the perturbation
$\phi$. But our application will satisfy the above interlacing condition, therefore,
we assume it.
Moreover,  
the eigenvalues satisfy
\begin{equation}\label{determinant:assumption2}
 \sum_{n=1}^\infty \big(\beta_n-\alpha_n\big)\le \beta_1-\alpha_1 + \sum_{n=2}^\infty \big(\alpha_{n-1}-\alpha_{n}\big)<\infty,  
\end{equation}
where the finiteness follows from the convergence of the sequence $(\alpha_n)_{n\in\N}$ to $0$. 

\begin{theorem}\label{prod:thm1}
Let $N\in\N$. We assume condition
 \eqref{determinant:assumption} to hold. Then,
 \begin{equation}
  \Big|  \det\Bigparens{\<\varphi_j, \psi_k\>}_{1\le j,k\le N}\Big|^2=\prod_{j=1}^N\prod_{k=N+1}^\infty\frac{\left|\beta_k-\alpha_j\right|\left|\alpha_k-\beta_j\right|}{\left|\alpha_k-\alpha_j\right|\left|\beta_k-\beta_j\right|}.
 \end{equation}
\end{theorem}

\begin{proof}[Proof of Theorem \ref{prod:thm1}] 
We use 
 the eigenvalue equations and assumption \eqref{determinant:assumption} to obtain for all $j,k\in\N$
\begin{equation}
 \<\varphi_j,\psi_k\>=\frac{\<\varphi_j,\phi\>\<\phi,\psi_k\>}{\beta_k-\alpha_j}.
\end{equation}
Hence, the multi-linearity of the determinant implies
\begin{align}
& \Big|  \det\Bigparens{\<\varphi_j, \psi_k\>}_{1\le j,k\le N}\Big|^2\notag\\
=& \Big|  \det\Bigparens{\frac{\<\varphi_j,\phi\>\<\phi,\psi_k\>}{\beta_k-\alpha_j}}_{1\le j,k\le N}\Big|^2\nonumber\\
=& \bigg(\prod_{j=1}^N\prod_{k=1}^N\big|\<\varphi_j,\phi\>\<\phi,\psi_k\>\big|^2\bigg) \Big|\det\Bigparens{\frac 1 {\beta_k-\alpha_j}}_{1\le j,k\le N}\Big|^2.\label{determinant}
\end{align}
Now, the remaining determinant can be computed explicitly. 
We use the Cauchy determinant formula to evaluate this,
see e.g. \cite[Lem. 7.6.A]{zbMATH06257273}, and end up with 
\begin{align}
\eqref{determinant}
= \bigg(\prod_{j=1}^N\prod_{k=1}^N\big|\<\varphi_j,\phi\>\<\phi,\psi_k\>\big|^2\bigg)
\frac{\prod_{j,k=1,j\neq k}^N\left|\beta_k-\beta_j\right|\left|\alpha_j-\alpha_k\right|}{\prod_{j,k=1}^N\left|\beta_k-\alpha_j\right|^2}.\label{det:eq3}
\end{align}
Corollary \ref{cor:det} below yields
\begin{align}
\eqref{det:eq3}
=&   
\bigg(\prod_{k=1}^N 
 \prod_{\substack{l=1\\l\neq k}}^\infty\frac{ \left|\alpha_l-\beta_k\right|}{ \left|\beta_l-\beta_k\right|}\bigg)
 \bigg(\prod_{j=1}^N 
 \prod_{\substack{l=1\\l\neq j}}^\infty\frac{ \left|\beta_l-\alpha_j\right|}{\left|\alpha_l-\alpha_j\right|}\bigg)\nonumber
 \prod_{\substack{j,k=1\\j\neq k}}^N\frac{\left|\beta_k-\beta_j\right|\left|\alpha_j-\alpha_k\right|}{\left|\beta_k-\alpha_j\right|^2}\\
 =&
 \prod_{j=1}^N\prod_{k=N+1}^\infty \frac{\left|\beta_k-\alpha_j\right|\left| \alpha_k-\beta_j\right|}{ |\beta_k-\beta_j||\alpha_j-\alpha_k|}.
\end{align}
This gives the claim, where we remark that by the finiteness of \eqref{determinant:assumption2} 
all products in the latter converge absolutely. 
\end{proof}

To complete the proof, we continue with computing the resolvents of the operators $A$ and $B$ in terms of their eigenvalues.

\begin{lemma}\label{determinant:lemma}
We assume 
\eqref{determinant:assumption}. 
Then, there exist $a,b\in \R$ with $ab=-1$ such that
 \begin{enumerate} 
  \item [(i)] for all $z\in\varrho(A)$
      \begin{equation}\label{prod:lemma1:eq1}
       \<\phi,\frac 1 {A-z}\phi\>+1= a\prod_{k=1}^\infty \frac{\beta_k-z}{\alpha_k-z}\,,
      \end{equation}
  \item [(ii)] for all $z\in\varrho(B)$
      \begin{equation}\label{prod:lemma1:eq2}
       \<\phi, \frac 1 {B-z}\phi\>- 1=  b\prod_{n=1}^\infty  \frac{\alpha_n-z}{\beta_n-z}\,.
       \end{equation}
 \end{enumerate}
\end{lemma}

\begin{corollary}\label{cor:det}
 Let $j,k\in\N$. Under the assumption 
 \eqref{determinant:assumption} 
 \begin{align}
   \left|\<\varphi_j,\phi\>\<\psi_k,\phi\>\right|^2
   =
   \left|\beta_j-\alpha_j\right|\left|\alpha_k-\beta_k\right|
   \bigg(\prod_{\substack{l=1\\l\neq j}}^\infty\frac{ \left|\beta_l-\alpha_j\right|}{\left|\alpha_l-\alpha_j\right|}\bigg)
   \bigg(\prod_{\substack{l=1\\l\neq k}}^\infty\frac{ \left|\alpha_l-\beta_k\right|}{ \left|\beta_l-\beta_k\right|}\bigg).
 \end{align}
\end{corollary}

\begin{proof}[Proof of Corollary \ref{cor:det}]
Using Lemma \ref{determinant:lemma} we compute the residue of the resolvents
\begin{align}
 \left|\<\varphi_j,\phi\>\right|^2&=\lim_{z\to\alpha_j} \left(\alpha_j-z\right)\<\phi,\frac 1 {A-z} \phi\>\nonumber\\
                                  &=\lim_{z\to\alpha_j} \left(\alpha_j-z\right)a\prod_{l=1}^\infty \frac{\left( \beta_l-z\right)}{\left(\alpha_l-z\right)}
                                  =a\left(\beta_j-\alpha_j\right)\prod_{\substack{l=1\\l\neq j}}^\infty
				    \frac{ \left(\beta_l-\alpha_j\right)}{\left(\alpha_l-\alpha_j\right)}\label{det:eq1}
\end{align}
and along the same line 
\begin{equation}\label{det:eq2}
 \left|\<\psi_k,\phi\>\right|^2=
 b\left(\alpha_k-\beta_k\right)
 \prod_{\substack{l=1\\l\neq k}}^\infty\frac{ \left(\alpha_l-\beta_k\right)}{ \left(\beta_l-\beta_k\right)}.
\end{equation}
Taking the absolute value and using $|ab|=1$, we get the result. 
\end{proof}

\begin{proof}[Proof of Lemma \ref{determinant:lemma}]
 First note that by the finiteness of \eqref{determinant:assumption2} the sequences
 \begin{equation}
  \bigg(\prod_{k=1}^N \frac{\beta_k-z}{\alpha_k-z}\bigg)_{N\in\N} 
  \quad \text{and} \quad \bigg(\prod_{n=1}^N \frac{\alpha_n-z}{\beta_n-z}\bigg)_{N\in\N}
 \end{equation}
 converge locally uniformly for all $z\in \varrho(A)\cap\varrho(B)$, see \cite[Thm. 252]{zbMATH00861508}.
 Therefore, the limits 
 \begin{equation}
  F(z):= \prod_{n=1}^\infty  \frac{\alpha_n-z}{\beta_n-z}\quad \text{and}\quad G(z):=\prod_{k=1}^\infty \frac{\beta_k-z}{\alpha_k-z}
 \end{equation}
 are well-defined analytic functions on $\varrho(A)\cap\varrho(B)$, 
 which fulfill $FG=1$. 
 Due to the locally uniform convergence, the derivative of $F$ satisfies
 \begin{align}
  F'(z) = &
  \lim_{N\to\infty}\sum_{l=1}^N \prod_{\substack{n=1 \\ n\neq l}}^N \frac{\alpha_n-z}{\beta_n-z} 
  \frac{\d }{\d z}  \frac{\alpha_l-z}{\beta_l-z} \notag\\
   = &
  \lim_{N\to\infty}\sum_{l=1}^N \prod_{\substack{n=1 \\ n\neq l}}^N 
  \frac{\alpha_n-z}{\beta_n-z} \frac{\alpha_l-\beta_l}{(\beta_l-z)^2}
   = 
  F(z) \lim_{N\to\infty}\sum_{l=1}^N \Big(\frac 1 {\beta_l-z} - \frac 1 {\alpha_l-z} \Big)\label{lemma:proof:eq5}
  \end{align}
  for all $z\in\varrho(A)\cap\varrho(B)$.
  We apply Lemma \ref{Feyn:Hell} below and obtain
  \begin{align}
   \eqref{lemma:proof:eq5}
   = &
  -F(z)\, \big\langle \frac 1 {A-z}\phi, \frac 1 {B-z}\phi \big\rangle.\label{lemma:proof:eq2}
 \end{align}
 Now, the resolvent identity implies for all $z\in \varrho(A)\cap \varrho(B)$
 \begin{equation}\label{lemma:eq:resolvent}
  \frac 1 {B-z} - \frac 1 {A-z} = - \frac 1 {A-z}\phi \big\langle \frac 1 {B-\bar z}\phi,\, \cdot \, \big\rangle
 \end{equation}
 which provides the equality
 \begin{equation}
  \frac 1 {A- z}\phi = \frac 1 {1- \langle \frac 1 {B-\bar z}\phi,\phi \rangle} \frac 1 {B- z}\phi.
 \end{equation}
 Inserting this into \eqref{lemma:proof:eq2}, we see that $F$ solves the differential equation
 \begin{equation}\label{prod:lemma1:eq3}
  F'(E)=  
  F(E) \frac 1 {\langle \phi, \frac 1 {B-E} \phi \rangle-1} \big\langle \phi, \Big(\frac 1 {B-E}\Big)^2\phi\big\rangle
 \end{equation}
 at least for all $E\in \varrho(A)\cap \varrho(B)\cap\R$.
 On the other hand
 the resolvent of $B$ is analytic in $\varrho(B)$ and
 the function  $t\mapsto\langle\phi, \frac 1 {B-t} \phi\rangle-1$, $t<0$,
  solves the above ODE \eqref{prod:lemma1:eq3} as well.
 Now, the general solution to this ODE is $f(t)=x_0\exp\Big( \int_{t_0}^t \d s\, \frac 1 {\langle \phi, \frac 1 {B-s} \phi \rangle-1} \big\langle \phi, \left(\frac 1 {B-s}\right)^2\phi\big\rangle\Big)$,
 for some initial condition $(t_0,x_0)$. Note that the functions $t\mapsto F(t)$ and $t\mapsto\langle\phi, \frac 1 {B-t} \phi\rangle-1$ are non-zero, thus
 $\langle\phi, \frac 1 {B-t} \phi\rangle-1=cF(t)$ for some $c\neq 0$. 
 This and the identity theorem for analytic functions give the claim.
Equation \eqref{prod:lemma1:eq1} follows from $F(z)G(z)=1$ and the identity
\begin{equation}\label{prod:lemma1:eq4}
 \Big(\big\langle\phi, \frac 1 {B-z} \phi\big\rangle-1\Big) \Big(\big\langle\phi,\frac 1 {A-z}\phi\big\rangle+ 1 \Big) = - 1 ,
\end{equation}
for all $z\in\varrho(A)\cap\varrho(B)$ 
which is a consequence of \eqref{lemma:eq:resolvent}.
\end{proof}

 \begin{lemma}\label{Feyn:Hell}
 Let $z \in \varrho(A)\cap \varrho(B)$ and assume \eqref{determinant:assumption}. Then, we obtain the following identity
  \begin{align}\label{Feyn:Hell10}
   \sum_{l=1}^\infty \Big(\frac 1 {\beta_l-z} - \frac 1 {\alpha_l-z} \Big)
   = - \big\langle \frac 1 {A - \bar z}\phi, \frac 1 {B-z}\phi \big\rangle.
  \end{align}
 \end{lemma}

 Let us point out that in the finite-dimensional case the above equality follows directly 
 from the resolvent equation, \eqref{lemma:eq:resolvent}.
 Nevertheless, the infinite-dimensional
 case is a bit more involved due to convergence issues.

 \begin{proof}
 For $\lambda\in\R$
 we define the operator 
 \begin{equation}
 A(\lambda):=A+\lambda|\phi\>\<\phi|
 \end{equation}
 and write $\alpha_l(\lambda)$ for the $l$th eigenvalue
 counted from above and $\varphi_l(\lambda)$ for the corresponding eigenvector. 
 Moreover, we remark that $\alpha_l(1)$ and $\varphi_l(1)$ correspond to $\beta_l$ and $\psi_l$. 
 Assumption \eqref{determinant:assumption} and the definite sign of the perturbation imply
 that the eigenvalues of $A(\lambda)$ are non-degenerate for all $\lambda\in[0,1]$.
 Thus, standard results, see
 \cite[Chap. XII]{MR0493421}, give differentiability of the eigenvalues
 for all $\lambda\in(0,1)$ and we apply the Feynman-Hellmann theorem, see e.g. \cite{MR968677}, 
 to deduce for all $l\in\N$ and $ \lambda\in (0,1)$
 \begin{equation}\label{Fey-Hell:eq2}
  \alpha_l'(\lambda) = |\<\varphi_l(\lambda),\phi\>|^2.
 \end{equation}
 Hence, we compute
 for $z\in\C\backslash\R$ 
 \begin{align}
  \sum_{l=1}^\infty \Big(\frac 1 {\beta_l-z} - \frac 1 {\alpha_l-z} \Big)
  = &
 -  \sum_{l=1}^\infty \int_0^1\d \lambda\, 
  \Big( \frac 1 {\alpha_l(\lambda) - z}\Big)^2 \alpha'_l(\lambda)\nonumber\\
  = &
 -  \sum_{l=1}^\infty \int_0^1\d \lambda\, \Big( \frac 1 {\alpha_l(\lambda) - z}\Big)^2
  \big|\big\<\varphi_l(\lambda),\phi\big\>\big|^2.\label{L}
  \end{align}
  The eigenvalue equation implies
  \begin{align}
   \eqref{L} = &
 -  \sum_{l=1}^\infty \int_0^1\d \lambda\, 
  \big\<\frac 1 {A(\lambda)-\bar z}\phi,\varphi_l(\lambda)\big\>\big\<\varphi_l(\lambda),\frac 1 {A(\lambda)-z}\phi\big\>\nonumber\\
  = &
  - \int_0^1\d \lambda\, \big\langle\phi, \Big(\frac 1 {A(\lambda) - z}\Big)^2 \phi\big \rangle,\label{Fey-Hell:eq3}
 \end{align}  
 where we used Fubini's theorem to interchange the integral with the sum and the fact that the vectors
 $\big(\varphi_l(\lambda)\big)_{l\in\N}$
 form an ONB. The resolvent identity \eqref{lemma:eq:resolvent} implies
 \begin{equation}\label{lemma:eq:resolvent2}
  \frac 1 {A(\lambda)-z} \phi = \frac 1 { 1+ \lambda \<\phi, \frac 1 {A-z}\phi\>} \frac 1 {A-z}\phi.
 \end{equation}
Therefore, we continue
\begin{align}
 \eqref{Fey-Hell:eq3}
 =&
 - \int_0^1\d \lambda\,\big\langle\phi, \Big(\frac 1 {A - z}\Big)^2 \phi\big \rangle
 \bigg( \frac 1 { 1+ \lambda \<\phi, \frac 1 {A-z}\phi\>}\bigg)^2\notag\\
 =&
 - 
 \big\langle\phi, \big(\frac 1 {A - z}\big)^2 \phi \big\rangle
 \int_0^1\d \lambda\, \frac{\d }{\d \lambda} \bigg(\frac 1 { 1+ \lambda \<\phi, \frac 1 {A-z}\phi\>}\bigg)
 \frac 1 {\langle\phi, \frac 1 {A - z} \phi \rangle}\notag\\
 =&
 \frac{\langle\phi, \big(\frac 1 {A - z}\big)^2 \phi \rangle}
      {\langle\phi, \frac 1 {A - z} \phi \rangle}
 \bigg(1 - \Big(\frac 1 { 1+ \<\phi, \frac 1 {A-z}\phi\>}\Big)\bigg)
 =
 - \frac
 {\langle\phi, \big(\frac 1 {A - z}\big)^2 \phi \rangle}
 { 1+ \<\phi, \frac 1 {A-z}\phi\>}.\label{Fey-Hell:eq30}
\end{align}
Equation \eqref{lemma:eq:resolvent2} with $\lambda=1$ provides the assertion for $z\in\C\backslash\R$, i.e.
\begin{equation}
 \eqref{Fey-Hell:eq30}
 = -
 \big\< \frac 1 {A-\bar z}\phi, \frac 1 {B-z}\phi \big\>.
\end{equation}
We note that both sides of \eqref{Feyn:Hell10} are continuous within $\varrho(A)\cap \varrho(B)$.
For the left hand side of \eqref{Feyn:Hell10}, this follows from assumption \eqref{determinant:assumption2} and
for the right hand side from the continuity of the resolvent. 
Therefore, we extend the result to all $z\in  \varrho(A)\cap \varrho(B)$.
 \end{proof}

\section{Proof of Theorem \ref{main:thm:3d}}\label{proof:thm:main}

We decompose the determinant according to the angular momentum decomposition 
\eqref{3-d:eq1}. 
This implies
\begin{equation}\label{proof:3d:eq2}
  \Big|  \det\Bigparens{\big\<\varphi_j^L, \psi_k^L\big\>}_{1\le j,k\le N_L(E)}\Big|^2= 
  \prod_{l\in\N_0}  \Big| \det\Bigparens{\big\<\varphi_j^L(\ell), \psi_k^L(\ell)\big\>}_{1\le j,k\le N^l_L(E)}\Big|^{2(2\ell+1)},
\end{equation}
where $\varphi_j^L(\ell)$ and $\psi_k^L(\ell)$ correspond to the radial part of the eigenfunctions 
lying in the $\ell$-th angular momentum channel
and $N_L^\ell(E)$ to the relative particle number in the $\ell$-th angular momentum channel.
More precisely, 
\begin{equation}\label{delta:rel-part-numb}
 N_L^\ell(E):=
              \#\big\{\, k\in\N:\ \exists\, j\in\{1,\cdots,N_L\}\ \text{with}\ \lambda_k^L(\ell)=\lambda_j^L \,\big\}
\end{equation}
 where $\left(\lambda_k^L(\ell)\right)_{k\in\N}$ denote the eigenvalues of $h_L^\ell$. 
Since we chose the eigenfunctions of $H_L$ and $H_{\alpha,L}$ to be the same in every angular momentum
channel $\ell\ge 1$ we obtain 
that only the $\ell=0$ term in the product \eqref{proof:3d:eq2} is different from $1$. Hence,  
\begin{equation}
 \Big|  \det\Bigparens{\big\<\varphi_j^L, \psi_k^L\big\>}_{1\le j,k\le N_L(E)}\Big|^2= 
 \Big|  \det\Bigparens{\big\<\varphi_j^L(0), \psi_k^L(0)\big\>}_{1\le j,k\le N^0_L(E)}\Big|^2.
\end{equation}

Thus, we reduced our problem to a problem on the half-axis, where the relative particle number satisfies

\begin{lemma}\label{lemma:thermolimitd3}
 Given $E>0$. Let $L$ and $N_L(E)\in\N$ such that $\frac {N_L(E)}{|B_L(0)|} \to \rho(E)$ as $L\to\infty$. 
 Then, 
 \begin{equation}
  \frac {N_L^0(E)}{L}\to \frac {\sqrt E}{\pi}=:\rho_0(E),
 \end{equation}
 as $L\to\infty$. 
\end{lemma}

\begin{proof} For any $E>0$
 \begin{equation}
  \displaystyle\lim_{L\to\infty} \frac{\#\{k:\lambda_k^L\le E\}}{|B_L(0)|}=\rho(E)=\lim_{L\to\infty} \frac {N_L(E)}{|B_L(0)|},
 \end{equation}
 where the first equality follows from
  e.g. \cite[Sct. XIII.15]{MR0493421}. 
 Hence, we obtain for an arbitrary $\epsilon>0$ the inequalities
 \begin{equation}
  \#\{\,k:\lambda_k^L\le E-\epsilon\} \le N_L(E) \le \#\{\,k:\lambda_k^L\le E+\epsilon\}
 \end{equation}
 for $L$ large enough. Since $\rho$ is 
 is strictly increasing, we obtain $\lambda^L_{N_L(E)}\to E$. 
 Therefore, $\lambda_{N^0_L(E)}^L(0)\to E$ as well because otherwise there would be a gap in the spectrum of $h^0$
 by the definition of the relative particle number $N^0_L(E)$. 
 This implies for an arbitrary $\epsilon>0$ and $L$ large enough 
 \begin{align}
  \Big|\frac{N_L^0(E)}{L}- \frac{\#\{\,k:\lambda_k^{L}(0)\le E\}}{L}\Big|
  &\le \Big|\frac{\#\{\,k:\left(\frac{k\pi}{L}\right)^2\in (E-\epsilon, E+\epsilon)\}}{L}\Big|\nonumber\\
  & \le \frac c {\sqrt E} \epsilon,
 \end{align}
 for some constant $c$. Since $\#\{\,k:\lambda_k^{L}(0)\le E\}/L\to \rho_0(E)$, as $L\to\infty$, 
 this yields the claim. 
\end{proof}

Given \eqref{proof:3d:eq2} and Lemma \ref{lemma:thermolimitd3}, Theorem \ref{main:thm:3d} will follow from

\begin{theorem}\label{main:thm:half-axis}
Let $E>0$. Then,
\begin{equation}
  \Big|  \det\Bigparens{\big\<\varphi^L_j(0), \psi^L_k(0)\big\>}_{1\le j,k\le N_L}\Big |^2
  =
  L^{-\zeta(E)+\oh(1)}
\end{equation}
as $L\to\infty$, $N_L\in\N$ and $N_L/L\to \frac{\sqrt E}\pi$,
where
\begin{equation}
 \zeta(E):= \frac 1 {\pi^2} \delta_\alpha^2(\sqrt E)
\end{equation}
and $\delta_\alpha$ is given by Definition \ref{delta:def:phase}. 
\end{theorem}

From now we shorten the notation and drop the $0$ and $L$-index of the eigenfunctions and eigenvalues.

Apart from the product representation discussed in Section \ref{sec:prod} 
the main ingredient to the proof of Theorem \ref{main:thm:half-axis} is a elementary formula expressing the
non-negative eigenvalues of the perturbed operator $h^0_{\alpha,L}$ in terms of the eigenvalues
of the operator $h^0_L$ plus corrections depending on the scattering phase shift $\delta_\alpha$. 
First, note that the eigenvalues of $h^0_L$ can be computed explicitly, see \cite{MR0493421}, i.e. for $n\in\N$
\begin{equation}\label{delta:eigenvalue-unperturb}
 \lambda_n=\left(\frac{n\pi} L\right)^2.
\end{equation}

\begin{lemma}\label{prod:le1}
Let $\delta_\alpha$ be given by Definition \ref{delta:def:phase}. Then,
\begin{enumerate}
 \item [(i)] for $\alpha\ge 0$ and $n\in\N$ the $n$th eigenvalues of $h^0_L$ and $h^0_{\alpha,L}$ satisfy
	\begin{equation}\label{delta:self-consistent}
	    0\le \sqrt{\mu_n} =\sqrt{\lambda_n} - \frac{\delta_\alpha(\sqrt{\mu_n})} L,
	\end{equation}
 \item [(ii)] for $\alpha\le 0$ and $n>1$ 	 
	      the $n$th eigenvalues of $h^0_L$ and $h^0_{\alpha,L}$ satisfy
	\begin{equation}\label{delta:self-consistent2}
	    0\le \sqrt{\mu_n} =\sqrt{\lambda_n} - \frac{\delta_\alpha(\sqrt{\mu_n})} L,
	\end{equation}
 \item [(iii)] and $\delta$ exhibits the following expansion
    \begin{equation}\label{delta:def:expansion}
     \delta_\alpha(\sqrt{\mu_n})
     =
     \delta_\alpha(\sqrt{\lambda_n}) - 
     \frac{\delta_\alpha'(\sqrt{\lambda_n}) \delta_\alpha(\sqrt{\lambda_n})} L + \oh\Big(\frac 1 L\Big),
    \end{equation}
    which is valid for all $\mu_n\ge 0$, and 
    the error term depends on $\alpha$ but is independent of $n$. 
\end{enumerate}
\end{lemma}

\begin{proof}
 Let $k>0$.
 Consider the eigenvalue problem
    \begin{equation}\label{delta:eq:finite-size}
      -u_k''= k^2 u_k, \qquad -4\pi\alpha u_k(0+) + u_k'(0+)=0.
    \end{equation}
    Introducing Pr\"ufer variables 
    \begin{equation}\label{delta:prufer}
     u_k(x)= \rho_u(x) \sin(\theta_k(x)) \qquad u_k'(x)=k \rho_u(x)\cos(\theta_k(x)),
    \end{equation}
    we see that any non-zero solution of \eqref{delta:eq:finite-size} is of the form
    \begin{equation}\label{eigenfct:Pruefer}
     u_k(x):= a\sin\Big(kx+\arctan\Big(\frac k {4\pi\alpha}\Big) \Big), 
    \end{equation}
    for some $0\neq a\in\C$.
    Since any eigenfunction $u_k$ to an eigenvalue $k^2$ of $h^0_{\alpha,L}$ is a solution of
    \eqref{delta:eq:finite-size} in $(0,L)$ and additionally satisfies $u_k(L-)=0$,
    we obtain that 
    \begin{equation}\label{delta:eq:finite-size:finite}
      u_{k}(L)=a\sin\Big(kL+\arctan\Big(\frac {k} {4\pi\alpha}\Big) \Big) = 0.
    \end{equation}
    On the other hand, all $k^2$ such that \eqref{delta:eq:finite-size:finite} is satisfied
    are eigenvalues of $h^0_{\alpha,L}$. 
    Since the function $k\mapsto kL+\arctan\left(\frac {k} {4\pi\alpha}\right) $ is strictly increasing
    we obtain for any $n\in\N$ a unique eigenvalue $\mu_n\ge 0$ of $h^0_{\alpha.L}$ such that 
     \begin{equation}\label{phaseshift-eq1}
     \sqrt{\mu_n}L+\arctan\Big(\frac {\sqrt{\mu_n}} {4\pi\alpha}\Big) = n\pi,
    \end{equation}
    where $\mu_1<\mu_2<\cdots$. 
    This proves (i). For the case $\alpha<0$ note that $h^0_{\alpha,L}$ admits a single negative eigenvalue.
    Therefore, \eqref{phaseshift-eq1} is only valid starting from the second eigenvalue of $h^0_{\alpha,L}$.
    This implies for all $n\in\N$
    \begin{equation}
     \sqrt{\mu_{n+1}}=  \sqrt{\lambda_n} - \frac{\arctan\big(\frac{\sqrt{\mu_{n+1}}}{4\pi\alpha}\big)}L 
     =
     \sqrt{\lambda_{n+1}} - \frac{\pi -\arctan\big(\frac{\sqrt{\mu_{n+1}}}{4\pi|\alpha|}\big)}L.
    \end{equation}
   (iii) follows directly from (i), (ii) and 
   Definition \eqref{delta:def:phase} from the phase shift. 
\end{proof}
\begin{corollary}\label{prod:delta:le1}
 The eigenvalues of $h^0_L$ and $h^0_{\alpha,L}$ satisfy
 \begin{equation}
  \mu_1 < \lambda_1 < \mu_2 < \lambda_2<\cdots.
 \end{equation}
\end{corollary}

\begin{proof}
 Note that $|\delta_\alpha(k)| < \pi $ for all $k>0$. 
 Thus, \eqref{delta:eigenvalue-unperturb} and \eqref{delta:self-consistent2}
 imply the corollary.
\end{proof}

Next we apply the results from Section \ref{sec:prod} to the determinant:

\begin{lemma}\label{prod:thm2}
Let $N\in\N$. Then,
 \begin{equation}
    \Big|  \det\Bigparens{\<\varphi_j, \psi_k\>}_{1\le j,k\le N}\Big|^2=
    \prod_{j=1}^N\prod_{k=N+1}^\infty\frac{\left|\mu_k-\lambda_j\right|\left|\lambda_k-\mu_j\right|}
    {\left|\lambda_k-\lambda_j\right|\left|\mu_k-\mu_j\right|}.
 \end{equation}
\end{lemma}

\begin{proof}
First note that $h^0_{\alpha,L}$ is bounded from below by \eqref{delta:lemma:rank-1}.
This and $h^0_L\ge 0$ 
imply
 $-E\in\rho(h_L)\cap\rho(h^0_{\alpha,L})$ for some $E>0$. Moreover,  \eqref{delta:lemma:rank-1} 
 provides
 \begin{equation}
  \frac 1  {h^0_L+E} - \frac 1 {h^0_{\alpha,L}+E}= \big|\eta_L^{E,\alpha}\big\>\big\<\eta_L^{E,\alpha}\big|,
 \end{equation}
 for some $\eta^L_E\in L^2((0,L))$  
and Corollary \ref{prod:delta:le1} gives
\begin{equation}
 \frac 1 {\mu_1+ E} > \frac 1 {\lambda_1+ E} > \frac 1 {\mu_2+ E} > \frac 1 {\lambda_2+ E} > \dotsb,
\end{equation}
the eigenvalues satisfy assumption \eqref{determinant:assumption2}. 
Furthermore, the operators $\frac 1 {h^0_L+E}$ and $\frac 1 {h^0_{\alpha,L}+E}$ are non-negative, have a kernel consisting of $0$ only and are compact. 
Therefore, we are in position to apply
 Theorem \ref{prod:thm1} and obtain
\begin{align}
  \Big|  \det\Bigparens{\<\varphi_j, \psi_k\>}_{1\le j,k\le N}  \Big|^2 &= 
  \prod_{j=1}^N\prod_{k=N+1}^\infty 
  \frac{\big|\frac 1 {\mu_k+E}-\frac 1 {\lambda_j+E}\big|\big|
  \frac 1 {\lambda_k+E}- \frac 1 {\mu_j+E}\big|}{\big|\frac 1 {\lambda_k+E}-\frac 1 {\lambda_j+E}
  \big|\big|\frac 1 {\mu_k+E} - \frac 1 {\mu_j+E}\big|}\nonumber\\
 &= 
 \prod_{j=1}^N\prod_{k=N+1}^\infty\frac{\left|\mu_k-\lambda_j\right|\left|\lambda_k-\mu_j\right|}
 {\left|\lambda_k-\lambda_j\right|\left|\mu_k-\mu_j\right|}.	  
\end{align}
\end{proof}

\begin{proof}[Proof of Theorem \ref{main:thm:half-axis}]
We start with the product representation given in Lemma \ref{prod:thm2}. 
Note that for $\alpha<0$ there is an ambiguity since there exists precisely one negative eigenvalue $\mu_1$. 
Therefore, we treat the $j=1$ term in the product separately. We define
\begin{equation}
 A_L^N:=
 \prod_{k=N+1}^\infty 
 \frac{\left|\mu_k-\lambda_1\right|\left|\lambda_k-\mu_1\right|}{\left|\lambda_k-\lambda_1\right|\left|\mu_k-\mu_1\right|}
 =
 \prod_{k=N+1}^\infty \Big|1+ 
 \frac{(\mu_k-\lambda_k)(\lambda_1-\mu_1)}{(\lambda_k-\lambda_1)(\mu_k-\mu_1)} \Big|
\end{equation}
and estimate using Corollary \ref{prod:delta:le1}
\begin{align}
 \sum_{k=N+1}^\infty \Big| \frac{(\mu_k-\lambda_k)(\lambda_1-\mu_1)}{(\lambda_k-\lambda_1)(\mu_k-\mu_1)} \Big|
 &\le |\lambda_1-\mu_1|\sum_{k=N+1}^\infty \frac{ \Big(\left(\frac{k\pi}L\right)^2 - \big(\frac{(k-1)\pi}L\big)^2\Big)}
						{ \left(\left(\frac{k\pi}L\right)^2 - \left(\frac{\pi}L\right)^2\right)
						  \left( \big(\frac{(k-1)\pi}L\big)^2 -\left(\frac{\pi}L\right)^2\right)}\nonumber\\
 &\le  
 \frac{L^2}{\pi^2}|\lambda_1-\mu_1|\sum_{k=N+1}^\infty \frac{(2k-1)}{(k^2-1)(k^2-2k)}\nonumber\\
 &\le
 c \Big(\frac{L}{N}\Big)^2 |\lambda_1-\mu_1|.
\end{align}
Since $h_L^\alpha$ is uniformly bounded from below with respect to $L$, see Lemma \ref{delta:unif:below},
\begin{equation}\label{j=1term}
 \ln A_L^N = \ln\bigg(\prod_{k=N+1}^\infty 
 \frac{\left|\mu_k-\lambda_1\right|\left|\lambda_k-\mu_1\right|}{\left|\lambda_k-\lambda_1\right|\left|\mu_k-\mu_1\right|}\bigg)
 = O(1)
\end{equation}
as $N,L\to\infty$ and $\frac N L\to \rho(E)>0$. 
Therefore, we are left with a product consisting of the non-negative eigenvalues and
 apply Lemma \ref{prod:thm2}, use Lemma \ref{prod:le1} (i) and 
 $\sqrt{\lambda_n}=\frac {n\pi} L$, $n\in \N$, to obtain
 \begin{align}
   &\ln\Big|  \det\Bigparens{\<\varphi_j, \psi_k\>}_{1\le j,k\le N}\Big|^2
   =
    \ln A_L^N\notag \\
    & +
   \sum_{j=2}^N\sum_{k=N+1}^\infty \ln\bigg( \frac{\big|\left(k\pi- \delta_\alpha(\sqrt{\mu_k}) \right)^2- 
   \left(j\pi\right)^2\big|\big|\left((k\pi)\right)^2-\left(j\pi  - 
   \delta_\alpha(\sqrt{\mu_j})\right)^2\big|}{\big|\left(k\pi\right)^2-\left(j\pi\right)^2\big|
   \big|\left(k\pi - \delta_\alpha(\sqrt{\mu_k})\right)^2- \left(j\pi- \delta_\alpha(\sqrt{\mu_j})\right)^2\big|} \bigg) 
   \label{prod:eq1}.
   \end{align}
   In the following the $\Oh(1)$ and $\oh(1)$ terms refer to the asymptotics $L,N\to\infty$, $N/L\to \rho_0(E)>0$. 
Equation \eqref{j=1term} above, Lemma \ref{det:lemma:e1} below 
and the abbreviation $g_k:=-\frac 1 \pi \delta_\alpha(\sqrt{\mu_k})$ for $k\in\N$ yield
\begin{align}
 \eqref{prod:eq1}
                   = -\sum_{j=2}^N\sum_{k=N+1}^\infty \frac{\big(2j g_j + g_j^2\big)\left(2k g_k+ g_k^2\right)}
                            {\big( \left(k  + g_k\right)^2- \left(j  +g_j \right)^2\big) \big(k^2-j^2\big)}
                            +\Oh\left(1\right).    \label{prod:eq3}
\end{align}
Using Lemma \ref{prod:le2} and the abbreviation $\delta_k:=-\frac 1 \pi \delta_\alpha(\sqrt{\lambda_k})$ for $k\in\N$, we have 
\begin{align}
 \eqref{prod:eq3}=-\sum_{j=2}^N\sum_{k=N+1}^\infty \frac{\big(2j \delta_j + \delta_j^2\big)\left(2k \delta_k+ \delta_k^2\right)}
                            {\big( \left(k  + \delta_k\right)^2- \left(j  +\delta_j \right)^2\big) \big(k ^2-j^2\big)}
                            +\Oh\left(1\right).  \label{prod:eq4}
\end{align}
Lemma \ref{prod:lemma:cut} implies
\begin{align}
 \eqref{prod:eq4}= -\sum_{j=2}^N\sum_{k=N+1}^{2N} \frac{4jk\delta_j\delta_k}{\big(k^2-j^2\big)^2}+ \Oh(1) .\label{prod:thm:eq1}
\end{align}
 Lemma \ref{prod:lemma_2} yields
\begin{align}
 \eqref{prod:thm:eq1} = -\frac 1 {\pi^2}\int_0^{\frac N L}\d x\, \int_{\frac{N+1} L}^{\frac{2N} L} \d y \,
 \frac{4xy\delta_\alpha(x\pi)\delta_\alpha(y\pi)}{(y^2-x^2)^2} + \Oh(1). \label{prod:thm:eq2}
\end{align}
We define for $0\le x< y$
\begin{equation}
g(x,y):= \frac{4xy\delta_\alpha(\pi x)\delta_\alpha(\pi y)}{(y+x)^2}
\end{equation}
The explicit representation of $\delta_\alpha$ implies for all $\epsilon>0$
\begin{equation}
\sup_{b>\epsilon}\, \sup_{(x,y)\in(0,b)\times (b,\infty)}\, \norm{(\nabla g)(x,y)}_2:=c(\epsilon)< \infty.
\end{equation}
Therefore, using the mean value theorem and the Cauchy-Schwarz inequality, we compute
for a $0<\epsilon< \sqrt E$ and $N,L$ big enough
\begin{align}
   &\int_0^{\frac N L} \d x \int_{\frac N L+ \frac 1 L}^{\frac{2N} L} \d y\,
   \Big|\frac{4xy\delta_\alpha(x\pi)\delta_\alpha(y\pi)}{(y+x)^2} - \delta_\alpha^2\left( N/L\right)\Big|\frac 1 {(y-x)^2}\notag\\
   \le & c(\epsilon) \int_0^{\frac N L} \d x \int_{\frac N L+ \frac 1 L}^{\frac{2N} L} \d y\,
   \bignorm{ \big(N/L- x, y-  N/ L\big)}_2 \frac 1 {(y-x)^2}\notag\\
   \le & 2c(\epsilon)  \int_0^{\frac N L} \d x \int_{\frac N L+ \frac 1 L}^{\frac{2N} L} \d y\,
   \frac 1{(y-x)} = \Oh(1),   \label{prod:thm:eq5}
\end{align}
where we used the inequality
\begin{equation}
\frac{\big| x - N/L\big| + \big| y - N/L\big| }{(y-x)^2}\le 2 \frac 1 {(y-x)},
\end{equation}
which is valid for all $x<N/L<y$. Moreover, since $\frac N L \to \frac{\sqrt E} \pi>0$, 
we compute
\begin{equation}
 \int_0^{\frac N L} \d x \int_{\frac N L+ \frac 1 L}^{\frac {2N} L} \d y \frac 1 {(y-x)^2} = \ln L + \Oh(1).\label{112}
\end{equation}
Hence, combining equation \eqref{prod:thm:eq5} and \eqref{112}, we end up with
\begin{align}
\eqref{prod:thm:eq2}&=  -\ln L \frac 1 {\pi^2} \delta_\alpha^2(\pi N/L) + \Oh(1)\label{1.}\\
                		       &=  -\ln L \frac 1 {\pi^2} \delta_\alpha^2(\sqrt E) + \oh(\ln L)\label{2.},
\end{align}
where the last line follows from $\pi \frac N L\to \sqrt E$.
This gives the assertion.
\end{proof}

\begin{appendix}

\section{Proof of the auxiliary lemmata}
In this section we prove the missing lemmata
used in the proof of Theorem \ref{main:thm:half-axis}. 
We do not claim to give optimal or very elegant estimates. 
Throughout this section we 
drop the index $\alpha$ in the scattering phase shift and
restrict ourselves to the case $\alpha<0$. 
This implies the following estimate on the phase shift
\begin{align}\label{aux:pf}
 \quad \delta(x)-\delta(y) \ge 0,
\end{align}
for $x<y$,
which we use in the sequel.
The case $\alpha\ge 0$ is even simpler since
in that case the
 Definition \eqref{delta:def:phase} of the phase shift implies
the uniform bound
\begin{equation}
\norm{\delta}_\infty\le \frac \pi 2,
\end{equation}
which simplifies some of the following estimates. 
Moreover, we use the elementary asymptotics
\begin{align}
 \displaystyle\sum_{j=1}^N\sum_{k=N+1}^\infty \frac 1 {(k-j)^2}=\Oh(\ln N),\label{RR}\\
 \displaystyle\sum_{j=1}^N\sum_{k=N+1}^\infty \frac 1 {(k-j)^\beta}=\Oh(1)\label{TT}
\end{align}
as $N\to\infty$, where $\beta>2$. 

\begin{lemma}\label{det:lemma:e1}
Set $g_k:=-\frac 1 \pi \delta(\sqrt{\mu_k})$ for $k\in\N$. Then,
 \begin{align}
     &\sum_{j=2}^N\sum_{k=N+1}^\infty \ln\bigg( \frac{\big(\left(k + g_k\right)^2- j^2\big)\big(k^2-\left(j + g_j\right)^2\big)}
	      {\big(\left(k + g_k \right)^2- \left(j + g_j\right)^2\big)\big(k^2-j^2\big)} \bigg)\label{det:le:eq1} \\
   = & -\sum_{j=2}^N\sum_{k=N+1}^\infty \frac{\big(2j g_j + g_j^2\big)\left(2k g_k+ g_k^2\right)}
                            {\big( \left(k +g_k\right)^2- \left(j +g_j\right)^2\big) \left(k^2-j^2\right)}
                            +\Oh(1)
 \end{align}
 as $N,L\to \infty$, $\frac N L\to \frac{\sqrt E} \pi$. 
\end{lemma}

\begin{proof}
We prove the assertion in two steps. First we consider the $j=N$ and $k=N+1$ summand. 
Note that Lemma \ref{lemma:thermolimitd3} above and $E>0$ imply
\begin{equation}\label{det:le:eq3}
 \lim_{\substack{N,L\to\infty\\ N/L\to \sqrt E/ \pi}} g_N = \lim_{\substack{N,L\to\infty\\ N/L\to \sqrt E/ \pi}} g_{N+1}
 = - \frac{\delta(\sqrt E)}\pi> -1.
\end{equation}
Thus, for $j=N$ and $k=N+1$
\begin{align}
& \lim_{\substack{N,L\to\infty\\ N/L\to \sqrt E/ \pi}} 
  \ln\bigg( \frac{\big(\left(N+1 + g_{N+1}\right)^2- N^2\big)\big((N+1)^2-\left(N + g_N\right)^2\big)}
	      {\big(\left(N+1 + g_{N+1} \right)^2- \left(N + g_N\right)^2\big)\big((N+1)^2-N^2\big)} \bigg)\nonumber\\
&=
\lim_{\substack{N,L\to\infty\\ N/L\to \sqrt E/ \pi}} 
\ln\bigg(\frac{\big(1+g_{N+1}\big)\big(1-g_N\big)}{\big(1+g_{N+1}-g_N\big)}\frac{\big(2N+1+g_{N+1}\big)\big(2N+1+g_N\big)}{\big(2N+1+g_{N+1}+g_N\big)\big(2N+1)}\bigg)\nonumber\\
&= \ln\Big( 1- \frac{\delta^2(\sqrt E)}{\pi^2}\Big).
\end{align}
Moreover, along the same line using \eqref{det:le:eq3}
\begin{align}
 \lim_{\substack{N,L\to\infty\\ N/L\to \sqrt E/ \pi}}
 -\frac{\left(2N g_N + g_N^2\right)\left(2(N+1) g_{N+1}+ g_{N+1}^2\right)}
                            {\big( \left(N+1 +g_{N+1}\right)^2- \left(N +g_N\right)^2\big) \big((N+1)^2-N^2\big)}
 = -\frac{\delta^2(\sqrt E)}{\pi^2}.                            
\end{align}
Therefore, the $j=N$ and $k=N+1$ term is of order $1$. 

For $j\le N<N+1<k$ we want to apply the bound
\begin{equation}\label{det:le:eq2}
 \big|\ln(1+x)-x\big|\le\, \frac {x^2} 2\frac 1 {1-|x|}
\end{equation}
for $x\in\R$ with $|x|<1$, to $x=x_{jk}$ where 
\begin{equation}
 x_{jk}:= -\frac{\big(2j g_j + g_j^2\big)\big(2k g_k+ g_k^2\big)}
                            {\big( \left(k +g_k\right)^2- \left(j +g_j\right)^2\big) \big(k^2-j^2\big)}.
\end{equation}
We estimate using
$|g_n|\le 1$ for all $n\in\N$ 
and $g_k-g_j\ge 0$
\begin{align}
 |x_{jk}|
 &\le 
 \Big|\frac{(2j+g_j)(2k+g_k)}{(j+g_j+k+g_k)(k+j)}\Big|\Big|\frac 1 {(k-j+g_k-g_j)(k-j)}\Big|\nonumber\\
 &\le 2 \frac 1 {(k-j)^2}.\label{T}
\end{align}
Since $j\le N<N+1<k$, this implies in particular $|x_{jk}|\le \frac 1 2$, and we continue using
\eqref{det:le:eq2} and \eqref{T}
\begin{align}
 \sum_{j=1}^N\sum_{k=N+2}^\infty \big| \ln(1+ x_{jk})- x_{jk}\big|
 &\le 
 \sum_{j=1}^N\sum_{k=N+2}^\infty x_{jk}^2\notag\\
 &\le    \sum_{j=2}^N\sum_{k=N+1}^\infty 
   4 \Big(\frac 1 {k-j}\Big)^4 = \Oh(1),
   \label{det:eq4}
\end{align}
as $N\to\infty$,
where we used \eqref{TT} in the last line. 
\end{proof}

\begin{lemma}\label{prod:le2}
Define $\delta_k:=-\frac 1 \pi \delta(\sqrt{\lambda_k})$ for $k\in\N$. Then,
\begin{align}
 \sum_{j=2}^N\sum_{k=N+1}^\infty \bigg|&\frac{\big(2j \delta_j+ \delta_j^2\big)\left(2k \delta_k+ \delta_k^2\right)}
                            {\big(\left(k +\delta_k \right)^2- \left(j +\delta_j\right)^2\big)}\nonumber\\
                       &- \frac{\big(2j g_j + g_j^2\big)\left(2k g_k+ g_k^2\right)}
                            {\big( \left(k +g_k\right)^2- \left(j + g_j \right)^2\big) }  \bigg|\,\frac 1{(k-j)^2}
                            = \oh(1)                   
\end{align}
as $N,L\to\infty$, $\frac N L\to \frac{\sqrt E}{\pi}$.
\end{lemma}

\begin{proof}
 First, using the expansion of Lemma \ref{prod:le1}, we obtain for all $n\in\N$, $n>1$,
 \begin{equation}\label{prod:lemma:eq1}
  |g_n-\delta_n|\le \frac 1 \pi \big|\delta(\sqrt{\mu_n}- \delta(\sqrt{\lambda_n})\big|\le 
  \frac{\norm{\delta}_\infty\norm{\delta'}_\infty  } {\pi L} := \frac{c} L,
 \end{equation}
  where the constant $c>0$ depends only on $\alpha$. 
  We prove the assertion in two steps. In the first step we consider the numerator only in the second step we consider the
  denominator. 
  Using \eqref{prod:lemma:eq1} we estimate
 \begin{align}
   &\sum_{j=2}^N\sum_{k=N+1}^\infty \bigg|\frac{\big(2j \delta_j+ \delta_j^2\big)\left(2k \delta_k+ \delta_k^2\right)-\big(2j g_j + g_j^2\big)\left(2k g_k+ g_k^2\right)}
                            {\big( \left(k +g_k\right)^2- \left(j + g_j \right)^2\big) \left(k^2-j^2\right)}
                         \bigg|\nonumber\\                  
                   \le &
                   \frac C L 
    \sum_{j=2}^N\sum_{k=N+1}^\infty \frac{(j+1)(k+1)}
    {\big( \left(k +g_k\right)^2- \left(j + g_j \right)^2\big) \left(k^2-j^2\right)}\nonumber\\
    \le &	
      \frac C L 
    \sum_{j=2}^N\sum_{k=N+1}^\infty \frac{(j+1)(k+1)}
    {\left(k+j-2\right)(k+j) \left(k-j\right)^2}
     =
     \Oh\Big(\frac{\ln N} L\Big) \label{prod:eq7}
 \end{align}
 as $N,L\to \infty$, $\frac N L\to\frac{\sqrt E}\pi$, where we used $|g_j+g_k|\le 2$, $g_k-g_j>0$ 
 for $j<k$ and
 \eqref{RR}.
 In order to estimate the denominator 
  we use  \eqref{prod:lemma:eq1} to obtain some constant $c>0$ independent of $j,k$ such that
 \begin{equation}
   \left| \left( \left(k +g_k\right)^2- \left(j + g_j \right)^2\right)  
   - \left( \left(k +\delta_k\right)^2- \left(j + \delta_j \right)^2\right) \right|\le c\frac{k+j} L. 
 \end{equation}
 Thus,
 \begin{align}
   &\sum_{j=2}^N\sum_{k=N+1}^\infty \left(2j \delta_j+ \delta_j^2\right)\left(2k \delta_k+ \delta_k^2\right)
                    \bigg|   \frac 1 {\big( \left(k +g_k\right)^2- \left(j + g_j \right)^2\big) \big(k^2-j^2\big)}\nonumber
                         \\   
                       &\qquad\qquad -\frac 1 {\big( \left(k +\delta_k\right)^2- \left(j + \delta_j \right)^2\big) \big(k^2-j^2\big)}\bigg|\nonumber\\
            \le &  \frac {4c} L\sum_{j=2}^N\sum_{k=N+1}^\infty \
		    \frac{jk(k+j)}
		    {\big(k^2-j^2\big)^2\big((k+g_k)^2 - (j+g_j)^2\big)\big(( k+\delta_k)^2- (j + \delta_j)^2\big)}\nonumber\\
	\le &
	\frac {4c} L\sum_{j=2}^N\sum_{k=N+1}^\infty 
	\frac {jk}{\big(k-j\big)^4\big(k+j-2\big)^2\big(k+j\big)} = \oh(1)\label{prod:eq8}
 \end{align} 
as $N,L\to\infty$ $N/L\to \frac{\sqrt E}\pi$, where we used $|g_k+g_j|\le 2$, $|\delta_k+\delta_j|\le 2$,
$g_k-g_j>0$ and $\delta_k-\delta_j>0$ for $j<k$. 
\end{proof}

\begin{lemma}\label{prod:lemma:cut}
The estimate
 \begin{align}
  \bigg|\sum_{j=2}^N\sum_{k=N+1}^\infty \frac{\big(2j \delta_j + \delta_j^2\big)\left(2k \delta_k+ \delta_k^2\right)}
                            {\big( \left(k  + \delta_k\right)^2- \left(j  +\delta_j \right)^2\big) \big(k^2-j^2\big)}                            
        -  \sum_{j=2}^N\sum_{k=N+1}^{2N} \frac {4jk\delta_j\delta_k}{(k^2-j^2)^2}\bigg|
     =\,   \Oh(1)\label{100}
 \end{align}
 holds 
 as $N,L\to\infty$, $\frac N L\to \frac{\sqrt E}{\pi}$.
\end{lemma}

\begin{proof}
 First, we bound the tail, i.e. 
 using
 $\delta_k-\delta_j>0$ for $k>j$ and $|\delta_n|\le 1$ for all $n\in\N$ we estimate
 \begin{align}
  &\sum_{j=2}^N\sum_{k=2N+1}^\infty \frac{\big(2j \delta_j + \delta_j^2\big)\left(2k \delta_k+ \delta_k^2\right)}
                            {\big( \left(k  + \delta_k\right)^2- \left(j  +\delta_j \right)^2\big) \left(k^2- j^2\right)}
 \le  \sum_{j=2}^N\sum_{k=2N+1}^\infty \frac {1}{(k-j)^2}\nonumber\\
 \le &  \sum_{k=2N+1}^\infty \frac N {(k-N)^2}=\Oh(1),\label{lemma:prod:eq1}
 \end{align}
 as $N\to\infty$.
 We insert 
 $\pm \displaystyle\sum_{j=2}^N\sum_{k=N+1}^{2N} \frac{4jk\delta_j\delta_k}{\big((k+\delta_k)^2-(j+\delta_j)^2\big)\big(k^2-j^2\big)}$ 
 in \eqref{100}.
Thus, in the next step $\delta_k-\delta_j>0$ yields 
  \begin{align}
  & \sum_{j=2}^N\sum_{k=N+1}^{2N}
  \bigg| \frac{\big(2j \delta_j + \delta_j^2\big)\left(2k \delta_k+ \delta_k^2\right)- 4jk\delta_j\delta_k}
               {\big( \left(k  + \delta_k\right)^2- \left(j  +\delta_j \right)^2\big) \big(k^2-j^2\big)}\bigg|\nonumber\\
  \le& \sum_{j=2}^N\sum_{k=N+1}^{2N}\bigg|                                             
  \frac{2(k+j)+1}
  {\big(k-j\big)^2\big(k+j\big)\big(k+j-2)\big)}\bigg|\nonumber\\
  \le& 3 \, \sum_{j=2}^N\sum_{k=N+1}^{2N}\bigg| \frac{1}{\big(k-j\big)^2\big(k+j-2\big)}\bigg|=\Oh\Big(\frac{\ln N} N\Big),
  \end{align}
  as $N\to\infty$, where we used \eqref{RR} in the last line. 
In the third step, again $|\delta_n|\le 1$ for $n\in\N$ yields
  \begin{align}
   &\sum_{j=2}^N\sum_{k=N+1}^{2N} \frac{4jk}{(k^2-j^2)} \, \bigg|\frac 1 {\big( \left(k  + \delta_k\right)^2- 
				  \left(j  +\delta_j \right)^2\big) }
								-\frac 1 {\big(k^2-j^2\big)}\bigg|\nonumber\\
   \le & \sum_{j=2}^N\sum_{k=N+1}^{2N} \frac{9jk\left(k+j\right)}
		    {\big(k^2-j^2\big)^2 \big(k+j-2\big)\big(k-j\big)}\nonumber\\
   \le & 9\, \sum_{j=2}^N\sum_{k=N+1}^{2N} \frac 1 {\big(k-j\big)^3}= \Oh(1),		    
  \end{align}
   as $N\to\infty$, where we used \eqref{TT}.
\end{proof}

\begin{lemma}\label{prod:lemma_2}
 The asymptotics
\begin{align}
 \bigg|\sum_{j=2}^N\sum_{k=N+1}^{2N} \frac{4jk\delta_j\delta_k}{\big(k^2-j^2\big)^2}-
 \frac 1 {\pi^2}\int_{0}^{\frac N L}\d x\,\int_{\frac{N+1} L}^{\frac{2N} L}\d y \,
 \frac{4xy\delta(x\pi)\delta(y\pi)}{\big(y^2-x^2\big)^2}\bigg|=\Oh(1)
\end{align}
holds  
 as $N,L\to\infty$, $\frac N L\to \frac{\sqrt E}{\pi}$.
\end{lemma}

\begin{proof}
 We recall that $\delta_k:=-\frac 1 \pi \delta(\sqrt{\lambda_k})$
 and we rewrite
 \begin{equation}
  \sum_{j=2}^N\sum_{k=N+1}^{2N} \frac{4jk\delta_j\delta_k}{\big(k^2-j^2\big)^2} 
  = 
  \frac 1 {L^2\pi^2} \sum_{j=2}^N\sum_{k=N+1}^{2N} 
  \frac{4\frac j L\frac k L\delta\big(\frac{j\pi} L\big)\delta\left(\frac{k\pi} L\right)}
  {\big(\big(\frac k L\big)^2- \big(\frac j L\big)^2\big)^2}.
 \end{equation}
Thus, we estimate
\begin{align}
&\bigg|\frac 1 {L^2} \sum_{j=2}^N\sum_{k=N+1}^{2N} 
\frac{\frac j L\frac k L\delta\big(\frac{j\pi} L\big)\delta\left(\frac{k\pi} L\right)}
		  {\left(\left(\frac k L\right)^2- \big(\frac j L\big)^2\right)^2} 
	 -\int_{\frac 1 L}^{\frac N L}\d x\,\int_{\frac{N+1} L}^{\frac{2N+1} L}\d y \,
	 \frac{xy\delta(x\pi)\delta(y\pi)}{\big(y^2-x^2\big)^2}\bigg|\nonumber\\
\le &
\sum_{j=2}^N\sum_{k=N+1}^{2N}  \int_{\frac{j-1} L}^{\frac j L}\d x\,\int_{\frac k L}^{\frac{k+1} L}\d y\, 
\Big|\, f \Big(\frac j L,\frac k L\Big)-f(x,y) \, \Big|\label{prod:lemma:eq2},
\end{align}
where
\begin{equation}
 f(x,y):=\frac {xy\delta(x\pi)\delta(y\pi)}{\big(y^2-x^2\big)^2}.
\end{equation}
Using the mean-value theorem and the Cauchy-Schwarz inequality we obtain
\begin{align}
 \eqref{prod:lemma:eq2}
 \le & 
 \sum_{j=2}^N\sum_{k=N+1}^{2N} 
 \sup_{(x,y)\in (\frac{j-1}{L},\frac j L)\times(\frac k L, \frac{k+1} L)}\bignorm{(\nabla f)(x,y)}_2\nonumber\\
&
 \qquad\qquad\qquad\times 
 \int_{\frac{j-1} L}^{\frac j L}\d x\,\int_{\frac k L}^{\frac{k+1} L}\d y \,
 \Bignorm{\Big(\frac j L -x,\frac k L-y\Big)}_2\nonumber\\
\le &
\frac 1 {L^3}\, \sum_{j=2}^N\sum_{k=N+1}^{2N} 
\sup_{(x,y)\in (\frac{j-1}{L},\frac j L)\times(\frac k L, \frac{k+1} L)}\bignorm{(\nabla f)(x,y)}_2,\label{prod:lemma_eq4}
\end{align}
where $|\cdot|_2$ denotes the Euclidean norm. We compute
\begin{align}
 &(\nabla f)(x,y)= \frac 1 {\big(y^2-x^2\big)^3}\\
		&\times \begin{pmatrix}
			  (y^2-x^2)(y\delta(x\pi)\delta(y\pi)+xy\delta'(x\pi)\delta(y\pi)\pi)+ 4x^2y\delta(x\pi)\delta(y\pi)\\
			  (y^2-x^2)(x\delta(x\pi)\delta(y\pi)+xy\delta(x\pi)\delta'(y\pi)\pi)-4 x y^2\delta(x\pi)\delta(y\pi)
		      \end{pmatrix}\nonumber\\
&=: \frac 1 {\big(y^2-x^2\big)^3}\,  g(x,y).		      
\end{align}
We estimate for $(x,y)\in \big(\frac{j-1} L, \frac j L\big)\times \left(\frac k L, \frac{k+1} L\right)$, $j\le N <k$,
\begin{equation}
\Big(\frac 1 {y^2-x^2}\Big)^3 \le \frac{L^6}{\big(k+j-1\big)^3\big(k-j\big)^3}\le \frac{L^6}{N^3}
\frac 1 {\big(k-j\big)^3}\label{101}
\end{equation}
and, using $\delta,\delta'\in L^\infty((0,\infty))$,
\begin{equation}
 \sup_{(x,y)\in (\frac{j-1}{L},\frac j L)\times(\frac k L, \frac{k+1} L)}\bignorm{g(x,y)}_2
 \le
 \sup_{(x,y)\in (0,\frac {2N+1} L)\times (0,\frac {2N+1} L)} \bignorm{g(x,y)}_2 =\Oh(1)\label{102}
\end{equation}
as $N,L\to\infty$, $\frac N L \to \frac{\sqrt E}{\pi}$.
Thus, \eqref{101} and \eqref{102} imply 
\begin{align}
 \eqref{prod:lemma_eq4}\le \Oh\bigg( \sum_{j=2}^N\sum_{k=N+1}^{2N} \frac 1 {(y-x)^3} \bigg)= \Oh(1)
\end{align}
as $N,L\to\infty$, $\frac N L\to \frac{\sqrt E}{\pi}$.
\end{proof}

\end{appendix}

\section*{Acknowledgement}
The author thanks Heinrich K\"uttler, 
Peter M\"uller, Peter Otte, Wolfgang Spitzer, and especially Alessandro Michelangeli for
fruitful discussions.

\newcommand{\etalchar}[1]{$^{#1}$}
\newcommand{\noopsort}[1]{}

\end{document}